\documentclass[mnsc]{informs2}              % for a regular run
\TheoremsNumberedThrough     % Preferred (Theorem 1, Lemma 1, Theorem 2)
\usepackage{latexsym}
\usepackage{url}
\usepackage{float}
\usepackage{texnansi}
\usepackage{color}
\usepackage{tikz}
\usepackage{subfigure}
\usepackage{enumerate}
\usepackage{algorithm}
\usepackage{algorithmic}
\usepackage[normalem]{ulem}
\usepackage{graphicx}
\usepackage{lmodern}
\usepackage{amsmath}
\usepackage{amssymb}
\usepackage{amsfonts}
\usepackage{graphicx}
\usepackage{natbib}
\usepackage{setspace}
\usepackage{mathrsfs}
\usepackage{multirow}
\floatstyle{ruled}
\usepackage{natbib}
 \bibpunct[, ]{(}{)}{,}{a}{}{,}%
 \def\bibfont{\small}%
\EquationsNumberedThrough    % Default: (1), (2), ...
\newcommand{\field}[1]{\ensuremath{\mathbb{#1}}}
\newcommand{\R}{\ensuremath{\field{R}}} % real numbers
\newcommand{\Rp}{\ensuremath{\R_+}} % positive real numbers
\newcommand{\Bscr}{\ensuremath{\mathcal B}}

\newcommand{\Mscr}{\ensuremath{\mathcal M}}
\newcommand{\Nscr}{\ensuremath{\mathcal N}}

\newcommand{\set}[1]{\left\{ #1 \right\}}

\newcommand{\abs}[1]{\lvert #1 \rvert}

\newcommand{\beps}{\varepsilon}

\newcommand{\defas}{\overset{\mathrm{def}}{=}}

\newcommand{\Mhatopt}{\hat{M}^{\mathrm{OPT}}}
\newcommand{\Mopt}{M^{\mathrm{OPT}}}
\newcommand{\Mhat}{\hat{\Mscr}}

\newcommand{\Ntilde}{\tilde{\Nscr}}
\newcommand{\uhat}{\hat{u}}

\newcommand{\exchangeouts}{\mathrm{\text{exchange-outs}}}
\newcommand{\exchange}{\mathrm{\text{exchange}}}
\newcommand{\add}{\mathrm{\text{add}}}
\newcommand{\Mtilde}{\tilde{\Mscr}}
\newcommand{\Rtilde}{\tilde{R}}
\newcommand{\Ctilde}{\tilde{C}}

\definecolor{commentcolor}{rgb}{0.5,0,0}

\begin{document}

% For new submissions, leave this number blank.
% For revisions, input the manuscript number assigned by the on-line
% system along with a suffix ".Rx" where x is the revision number.
\MANUSCRIPTNO{}

%\setcounter{page}{1}
%\DOI{10.1287/xxxx.0000.0000}%


\RUNTITLE{Assortment Optimization Under a Nonparametric Choice Model}

\TITLE{\textsf{
\title{Assortment Optimization Under a Nonparametric Choice Model}}}

\ARTICLEAUTHORS{%
\AUTHOR{Vivek F. Farias}
\AFF{MIT Sloan, \EMAIL{vivekf@mit.edu}} %, \URL{}}
\AUTHOR{Srikanth Jagabathula}
\AFF{EECS, MIT, \EMAIL{jskanth@alum.mit.edu}}
\AUTHOR{Devavrat Shah}
\AFF{EECS, MIT, \EMAIL{devavrat@mit.edu}}
} % end of the block

\SingleSpaced

\ABSTRACT{%
  We consider the problem of static assortment optimization, where the
  goal is to find the assortment of size at most $C$ that maximizes
  revenues. This is a fundamental decision problem in the area of
  Operations Management.  It has been shown that this problem is
  provably hard for most of the important families of parametric of
  choice models, except the multinomial logit (MNL) model. In
  addition, most of the approximation schemes proposed in the
  literature are tailored to a specific parametric structure. We
  deviate from this and propose a general algorithm to find the
  optimal assortment assuming access to only a subroutine that gives
  revenue predictions; this means that the algorithm can be applied
  with any choice model. We prove that when the underlying choice
  model is the MNL model, our algorithm can find the optimal
  assortment efficiently.  }

%\KEYWORDS{deterministic inventory theory; infinite linear programming duality;
%  existence of optimal policies; semi-Markov decision process; cyclic schedule}

% \KEYWORDS{butter, margarine, silliness} \HISTORY{This paper was
% first submitted on April 12, 1922 and has been with the authors for
% 83 years for 65 revisions.}

%\chapter{Decision problems: assortment
%optimization} \label{chap:assort_opt}
\maketitle
\section{Introduction}
This paper deals with the application of choice models to make decisions. There
are several important practical applications where the end-goal is to make a
decision, and a choice model is a critical component to making that
decision. The main application area of our focus is the set of decision problems
faced by operations managers. In this context, a central decision problem is the
{\em static assortment optimization} problem in which the goal is to find the
optimal assortment: the assortment of products with the maximum revenue subject
to a constraint on the size of the assortment. Solving the decision problem
requires two components: (a) a subroutine that uses historical sales transaction
data to predict the expected revenues from offering each assortment of products,
and (b) an optimization algorithm that uses the subroutine to find the optimal
assortment. This paper deals with desigining an efficient optimization algorithm.

As one can imagine, the problems of predicting revenues and finding the optimal
assortment are important in their own right, and their consideration is
motivated by the fact that any improvements to existing solutions will have
significant practical implications. Specifically, solutions to these two
problems lead to a solution to the {\em single-leg, multiple fare-class yield
  management problem}; this problem is central to the area Revenue Management
(RM) and deals with the allocation of aircraft seat capacity to multiple fare
classes when customers exhibit choice behavior. In particular, consider an
airline selling tickets to a single-leg aircraft. Assume that the airline has
already decided the fare classes and is trying to dynamically decide which
fare-classes to open as a function of the remaining booking time and the
remaining number of seats. This dynamic decision problem can be cast in a
reasonably straightforward manner as a dynamic program with one state
variable. As shown in~\cite{Talluri04}, the solution to the dynamic
program reduces to solving a slight variant of the static assortment
optimization problem. Thus, solution to the two problems effectively solves the
single-leg, multiple fare-class yield management problem — a central problem to
RM with huge practical implications.

Given the subroutine to predict revenues, we need an efficient algorithm to
search for the optimal assortment. In particular, we are interested in solving
\begin{equation*}
  \argmax_{\abs{\Mscr} \leq C}~~R(\Mscr),
\end{equation*}
where $R(\Mscr)$ is the expected revenue from offering assortment $\Mscr$. In
this chapter, we assume access to a subroutine that can efficiently generate
revenue predictions for each assortment $\Mscr$, and our goal is to design an
optimization algorithm that minimizes the number of calls to the subroutine. The
revenue predictions can themselves be generated either using a specific
parametric choice model or using the nonparametric approach described in the
previous chapter. Assuming there are $N$ products and a constraint of $C$ on the
size of the optimal assortment, exhaustive search would require $O(N^C)$ calls
to the revenue subroutine. Such an exhaustive search is prohibitive in practice
whenever $N$ or $C$ is large. Therefore, our goal is to propose an algorithm
that can produce a ``good'' approximation to the optimal assortment with only a
``few'' calls to the revenue subroutine. Existing approaches focus on exploiting
specific parametric structures of choice models to solve the decision problem
efficiently. In this context,~\cite{Rus10} have proposed an efficient algorithm
to find the optimal assortment in $O(NC)$ operations whenever the underlying
model is the MNL model. Unfortunately, beyond the simple case of the MNL model,
the optimization problem or its variants are provably hard (like the NL and MMNL
models; see~\cite{Rus09} and \cite{Rus10mmnl}). In addition, the algorithms
proposed in the literature (both exact and approximate) heavily exploit the
structure of the assumed choice model; consequently, the existing algorithms --
even without any guarantees -- cannot be used with other choice models like the
probit model or the mixture of MNL models with a continuous mixture. Given these
issues, our goal is to design a general optimization scheme that is (a) not
tailored to specific parametric structures and (b) requires only a subroutine
that gives revenue estimates for assortments.

{\bf Overview of our approach.} We propose a general set-function optimization
algorithm, which given a general function defined over sets, finds an estimate
of the set (or assortment) where the function is maximized. This set-function
optimization algorithm clearly applies to the static assortment optimization
problem, thereby yielding the optimization scheme with the desired
properties. Note that since we are considering a very general setup, there is
not much structure to exploit. Hence, we adopt the greedy method -- the general
technique for designing heuristics for optimization problems. However, a naive
greedy implementation algorithm fails even in the simple case of the MNL
model. Specifically, consider the simpler un-capacitated decision problem. Here,
a naive greedy implementation would start with the empty set and incrementally
build the solution set by adding at each stage a product that results in the
maximum increase in revenue; this process would terminate when addition of a
product no longer results in an increase in revenue. It is easy to see that the
naive implementation would succeed in solving the decision problem only if the
optimal assortments exhibit a nesting property: the optimal assortment of size
$C_1$ is a subset of the optimal assortment of size $C_2$ whenever $C_1 <
C_2$. Unfortunately, the nesting property does not hold even in the case of the
MNL model. In order to overcome this issue, we allow for greedy ``exchanges'' in
addition to greedy ``additions.''  Particularly, at every stage, we allow a new
product to be either added (which we call an ``addition'') to the solution set
or replace an existing product (which we call an ``exchange'') in the solution
set; the operation at each stage is chosen greedily. The termination condition
now becomes an interesting question. As in the naive implementation, we could
terminate the process when addition or exchange no longer results in an increase
in revenue. However, since we never run out of products for exchanges, the
algorithm may take an exponential (in the number of products) number of steps to
terminate. We overcome this issue by introducing a control parameter that caps
the number of times a product may be involved in exchanges. Calling that
parameter $b$, we show that the algorithms calls the revenue subroutine $O(N^2 b
C^2)$ times for the capacitated problem. We thus obtain a general algorithm with
the desired properties to solve the static assortment optimization problem.

{\bf Guarantees for our algorithm.} We derive guarantees to establish the
usefulness of our optimization procedure. For that, we first consider the case
of the MNL model, where the decision problem is well-understood. Specifically,
we assume that the underlying choice model is an instance of the MNL family and
the revenue subroutine yields revenue estimates for assortments under the
specific instance. We can show that the the algorithm we propose, when run with
$b \geq C$, succeeds in finding the optimal assortment with $O(N^2C^3)$ calls to
the revenue subroutine. Therefore, in the special case when the underlying
choice model is the MNL model, our algorithm captures what is already known. It
also provides a simpler alternative to the more complicated algorithm proposed
by \cite{Rus10}. We also consider the case when noise corrupts the available
revenue estimates -- a common practical issue. In this case, we show that our
algorithm is robust to errors in the revenue estimates produced by the
subroutine. Particularly, if the underlying choice model is the MNL model and
the revenue estimate produced by the subroutine may not be exact but within a
factor $1 - \beps$ of the true value, then we can show that our algorithm finds
an estimate of the optimal assortment with revenue that is within $1 - f(\beps)$
of the optimal value; here $f(\beps)$ goes to zero with $\beps$ and also depends
on $C$ and the parameters of the underlying model. In summary, our theoretical
analysis shows that our algorithm finds the exact optimal solution in the
noiseless case or a solution with provable guarantees in the noisy case,
whenever the underlying choice model is the MNL model. In this sense, our
results subsume what is already known in the context of the MNL model.

In the context of the more complicated models like the nested logit (NL) and the
mixtures of MNL models, the decision problem is provably hard. As discussed
above, even obtaining a PTAS can be very complicated and requires careful
exploitation of the structure. We however believe that it is possible to obtain
``good'' approximations to the optimal assortments in practice. 

{\bf Organization.} Next, we describe in detail the optimization algorithm we
propose and the guarantees we can provide. The rest of the chapter is organized
as follows. The optimization algorithm, which we call {\sc GreedyOPT} is
described in Section~\ref{sec:greedyopt_description}. We then describe the
precise guarantees we can provide on the algorithm in
Section~\ref{sec:greedyopt_mainresults}. Finally, we present the proofs of our
results in Section~\ref{sec:greedyopt_proofs} before concluding in
Section~\ref{sec:greedyopt_conclusion}.

\section{Description of {\sc GreedyOPT}} \label{sec:greedyopt_description}
We now provide the detailed description of our optimization algorithm {\sc
  GreedyOPT}. As noted above, most of the algorithms proposed in the literature
-- both exact and approximate -- are based on heavily exploiting the structure
of the assumed choice model. Unfortunately, since we are considering a very
general setup, there is not much structure to exploit. Hence, we adopt the
greedy method -- the general technique for designing heuristics for optimization
problems. 

A naive greedy implementation however fails even in the simple case of the MNL
model. Specifically, consider the simpler un-capacitated decision problem. Here,
a naive greedy implementation would start with the empty set and incrementally
build the solution set by adding at each stage a product that results in the
maximum increase in revenue; this process would terminate when addition of a
product no longer results in an increase in revenue. It is easy to see that the
naive implementation would succeed in solving the decision problem only if the
optimal assortments exhibit a nesting property: the optimal assortment of size
$C_1$ is a subset of the optimal assortment of size $C_2$ whenever $C_1 <
C_2$. Unfortunately, the nesting property does not hold even in the case of the
MNL model. 

In order to overcome this issues associated with the naive greedy
implementation, we allow for greedy ``exchanges'' in addition to greedy
``additions.''  Particularly, at every stage, we allow a new product to be
either added (which we call an ``addition'') to the solution set or replace an
existing product (which we call an ``exchange'') in the solution set; the
operation at each stage is chosen greedily. The termination condition now
becomes an interesting question. As in the naive implementation, we could
terminate the process when addition or exchange no longer results in an increase
in revenue. However, since we never run of products for exchanges, the algorithm
may take an exponential (in the number of products) number of steps to
terminate. We overcome this issue by introducing a control parameter that caps
the number of times a product may be involved in exchanges. Calling that
parameter $b$, we show that the algorithms calls the revenue subroutine $O(N^2 b
C^2)$ times for the capacitated problem. We thus obtain a general algorithm with
the desired properties to solve the static assortment optimization problem.

The formal description of the algorithm is provided in
Figures~\ref{fig:greedy-opt} and \ref{fig:add-x}. For convenience, whenever an
exchange takes place, we call the product that is removed as the product that is
{\em exchanged-out} and the product that is introduced as the product that is
{\em exchanged-in}. Now, the algorithm takes as inputs the capacity $C$, the
initial assortment size $S$, and a bound $b$ on the number of exchange-outs. The
algorithm incrementally builds the solution assortment. Specifically, it
searches over all assortments of size $S$. For each such assortment, the
algorithm calls the subroutine {\sc GreedyADD-EXCHANGE} (formally described in
Figure~\ref{fig:add-x}) at most $C - S$ times to construct an assortment of size
at most $C$. Of all such constructed assortments, the algorithm returns the one
with the maximum revenue. 

\begin{figure}[!h]
\label{fig:greedy-opt}
\caption{{\sc GreedyOPT}}
\vspace{0.2in}
%\hrule
\vspace{3pt}
\noindent\framebox[\linewidth][t]{
  \begin{minipage}{0.97\linewidth} {\noindent \bf Input:} Initial size $S$,
    capacity constraint $C$ such that $1 \leq S \leq C \leq N$, and revenue
    function $R(\cdot)$.
\vspace{5pt}

{\noindent \bf Output:} Estimate of optimal assortment $\Mhatopt$ of
size $\abs{\Mhatopt} \leq C$

\vspace{12pt}

\noindent{\bf Algorithm:} \\
%\hrule
{\em Initialization:} $\Mhatopt \leftarrow \emptyset$\\
{\bf for each }$\Mscr \subset \Nscr$ such that $\abs{\Mscr} = S$ \\
\hspace*{0.5cm} {\bf for } $S+1 \leq i \leq C$\\
\hspace*{1cm} $\Mscr \leftarrow${\sc GreedyADD-EXCHANGE}$(\Mscr,
\Nscr, b, R(\cdot))$\\
\hspace*{0.5cm} {\bf end for}\\
\hspace*{0.5cm} {\bf if }$R(\Mhatopt) < R(\Mscr)$\\
\hspace*{1cm} $\Mhatopt \leftarrow \Mscr$\\
\hspace*{0.5cm} {\bf end if}\\
{\bf end for} \\
{\em Output:} $\Mhatopt$\\
\vspace{3pt}
\end{minipage}}
%\hrule
\end{figure}

\begin{figure}[!h]
\label{fig:add-x}
\caption{{\sc GreedyADD-EXCHANGE}}
\vspace{0.2in}
%\hrule
\vspace{3pt}
\noindent\framebox[\linewidth][t]{
  \begin{minipage}{0.97\linewidth} {\noindent \bf Input:} assortment $\Mscr$,
    product universe $\Nscr$, revenue function $R(\cdot)$, maximum number of
    exhange-outs $b$
\vspace{5pt}

{\noindent \bf Output:} Estimate of optimal assortment of size at most $\abs{\Mscr} + 1$

\vspace{12pt}

\noindent{\bf Algorithm:} \\
%\hrule
{\em Initialization:} $\hat{\Mscr} \leftarrow \Mscr$, $\tilde{\Nscr}
\leftarrow \Nscr$, $\exchangeouts(i) = 0$ for each $i \in \Nscr$\\
{\bf while } $\tilde{\Nscr} \neq \emptyset$\\
\hspace*{0.5cm} \textcolor{commentcolor}{//try exchanging products} \\
\hspace*{0.5cm} $i^*, j^* = \argmax_{i \in \hat{\Mscr},  j \in \tilde{\Nscr}}
R\left( (\hat{\Mscr} \setminus \set{i}) \cup \set{j} \right)$\\
\hspace*{0.5cm} $\Mtilde_{\exchange} \leftarrow (\hat{\Mscr} \setminus
\set{i}) \cup \set{j}$\\\\
\hspace*{0.5cm} \textcolor{commentcolor}{// try adding a product }\\
\hspace*{0.5cm} $k^* = \argmax_{k \in \tilde{\Nscr}} R(\hat{\Mscr}
\cup \set{k})$\\
\hspace*{0.5cm} $\Mtilde_{\add} \leftarrow \hat{\Mscr} \cup
\set{k^*}$\\\\
\hspace*{0.5cm} {\bf if} $\abs{\Mhat} < \abs{\Mscr} +1$ {\bf and}
$R(\Mtilde_{\add} > R(\Mscr)$ {\bf and} $R(\Mtilde_{\add}) >
R(\Mtilde_{\exchange})$\\
\hspace*{1cm} \textcolor{commentcolor}{// add the product $k^*$} \\
\hspace*{1cm} $\Mhat \leftarrow \Mtilde_{\add}$\\
\hspace*{1cm} $\Ntilde \leftarrow \Ntilde \setminus \set{k^*}$\\
\hspace*{0.5cm} {\bf else if} $R(\Mtilde_{\exchange}) > R(\Mscr)$\\
\hspace*{1cm} \textcolor{commentcolor}{// exchange products $i^*$ and $j^*$}\\
\hspace*{1cm} $\Mhat \leftarrow \Mtilde_{\exchange}$\\
\hspace*{1cm} $\exchangeouts(i^*) \leftarrow \exchangeouts(i^*) + 1$\\
\hspace*{1cm} {\bf if} $\exchangeouts(i) \geq b$\\
\hspace*{1.5cm} $\Ntilde \leftarrow \Ntilde \setminus \set{j^*}$\\
\hspace*{1cm} {\bf else}\\
\hspace*{1.5cm} $\Ntilde \leftarrow \left(\Ntilde \setminus
  \set{j^*}\right) \cup \set{i^*}$\\
\hspace*{0.5cm} {\bf else} \\
\hspace*{1cm} {\bf break from while} \\
\hspace*{0.5cm} {\bf end if}\\
{\bf end while}\\
{\em Output:} $\hat{\Mscr}$\\
\vspace{3pt}
\end{minipage}}
%\hrule
\end{figure}

\noindent{\bf Running-time complexity:} It is easy to see that the number of
times {\sc GreedyOPT} calls the revenue function $R(\cdot)$ is equal to
$(C-S)\binom{N}{S}$ times the number of times {\sc GreedyADD-EXCHANGE} calls the
revenue function. In order to count the number of times {\sc GreedyADD-EXCHANGE}
calls the revenue function $R(\cdot)$, we first count the number of times the
while loop in {\sc GreedyADD-EXCHANGE} is executed. The number of times the
while loop runs is bounded above by the maximum number of iterations before the
set $\tilde{\Nscr}$ becomes empty. In each iteration either an addition or an
exchange takes place. Since there is at most one addition that can take place
and $\abs{\Ntilde}$ decreases by $1$ whenever $\exchangeouts(i)$ of a product
$i$ reaches $b$, it follows that the while loop runs for at most $Nb + 1$
iterations. In each iteration of the while loop, the revenue function is called
at most $O(CN)$ times. Thus, {\sc GreedyADD-EXCHANGE} calls the revenue function
at most $O(CbN^2)$ times. Since $\binom{N}{S} = O(N^S)$, we can now conclude
that {\sc GreedyOPT} calls the revenue function $O(C^2 b N^{S+2})$. The choice
of $S$ will depend on the accuracy of revenue estimates we have access to. Next,
we provide guarantees on {\sc GreedyOPT}, which provide guidance on the choice
of $S$.

\section{Theoretical guarantees for {\sc
    GreedyOPT}} \label{sec:greedyopt_mainresults} We now give a precise
description of the main results we can establish for the {\sc GreedyOPT}
algorithm. Specifically, suppose that the underlying choice model is an MNL
model with weights $w_0 = 1$ for product $0$ and $w_i$ for product $i \in
\Nscr$; recall that the choice probabilities are given by
\begin{equation*}
  \mathbb{P}(i \vert \Mscr) = \frac{w_i}{1 + \sum_{j \in \Mscr} w_j}.
\end{equation*}
Note that $1$ appears in the denominator because of the no-purchase option. In
particular, the probability that an arriving customer leaves without purchasing
anything when assortment $\Mscr$ is on offer is given by 
\begin{equation*}
  \mathbb{P}(0 \vert \Mscr) = \frac{1}{1 + \sum_{i \in \Mscr} w_i}.
\end{equation*}
Let $R(\Mscr)$ denote the expected revenue from assortment $\Mscr$. Under the
MNL model, we have
\begin{equation*}
  R(\Mscr) = \frac{\sum\limits_{i \in \Mscr} p_i w_i}{1 + \sum\limits_{i \in
      \Mscr} w_i},
\end{equation*}
where $p_i$ is the price or the revenue obtained from the sale of product $i$. 

We now have the following theorem when the revenue subroutine provides exact
revenues:
\begin{theorem} \label{thm:mnl-exact} 

  Suppose the underlying model is the MNL model with weights $w_1, w_2, \dotsc,
  w_N$ and the revenue subroutine provides exact revenues. Then, for any $S \geq
  0$ and $b \geq C+1$, the {\sc GreedyOPT} algorithm finds the optimal solution
  to {\sc Capacitated OPT} problem.

\end{theorem}

Therefore, taking $S=0$ and $b=C+1$, {\sc GreedyOPT} finds the optimal
assortment of size at most $C$ by calling the revenue function
$O(N^2C^3)$. Thus, our algorithm provides a simpler alternative to the more
complicated algorithm proposed by \cite{Rus10}.

We next show that the {\sc GreedyOPT} algorithm is robust to errors in the
available revenue estimates. Specifically, we consider the more realistic
setting where one has access to only approximate estimates of revenues i.e., we
assume access to a function $\Rtilde(\cdot)$ such that for any assortment
$\Mscr$ we have
\begin{equation*}
  (1 - \beps(\Mscr)) R(\Mscr) \leq \Rtilde(\Mscr) \leq R(\Mscr)
\end{equation*}
for some parameter $0 < \beps(\Mscr) < 1$. Naturally, the parameter
$\beps(\Mscr)$ determines the quality of revenue estimates we have
available. Assuming that we have access to only approximate revenues, we find
the optimal assortment by running {\sc GreedyOPT} with approximate revenues. In
order to describe the result, we need some notation. For any assortment $\Mscr$,
let $w(\Mscr)$ denote $1 + \sum_{i \in \Mscr} w_i$. Further, let 
\begin{equation*}
  \beps_{\max} \defas \max_{\Mscr \colon \abs{\Mscr} \leq C} \beps(\Mscr) \quad
  \text{and} \quad W_C^{\max} \defas \max_{\Mscr \colon \abs{\Mscr} \leq C} w(\Mscr).
\end{equation*}

Finally, we defer to the next section the precise definitions of two quantities $\bar{C}(\delta_C)$
and $\delta_C$ that we need to describe the theorem; it suffices to say that as
$\beps_{\max} \to 0$, we have $\delta_C \to 0$ and $\bar{C}(\delta_C) \to C$.

% For any assortment $\Mscr$, let
% \begin{equation*}
%   \delta(\Mscr) \defas w(\Mscr) \frac{\beps(\Mscr)}{1 - \beps(\Mscr)}.
% \end{equation*}
% Also, let 
% \begin{equation*}
%   \delta_C \defas \max_{\Mscr \colon \abs{\Mscr} \leq C} \delta(\Mscr).
% \end{equation*}
% For any $\delta > 0$ and $u \in \R$, let 
% \begin{equation*}
%   i_S(u)  \defas \min_{i \in B_S(u)} h_i(u).
% \end{equation*}
% Moreover, let 
% \begin{equation*}
%   \bar{B}_S(\delta, u) \defas B_S(u) \cup \set{j \in \Nscr\setminus B_S(u) \colon
%     h_{i_S(u)}(u) - h_j(u)\leq \delta u},
% \end{equation*}
% and
% \begin{equation*}
%   \bbar{B}_S(\delta, u) \defas \set{j \in B_S(u) \colon
%     h_j(u) - h_{i_S(u)}(u) > \delta u}.
% \end{equation*}

% Also, let
% \begin{equation*}
%   \bar{C}(\delta) \defas \max_{u \in \Rp} \abs{\bar{B}_S(\delta, u)},
% \end{equation*}
% \begin{equation*}
%   \bbar{C}(\delta) \defas \min_{u \in \Rp} \abs{\bbar{B}_S(\delta, u)}
% \end{equation*}
% and
% \begin{equation*}
%   c(\delta) \defas C - \bbar{C}(\delta).
% \end{equation*}

% Note that according to our definitions, it is easy to see that
% $\bbar{B}_S(\delta, u) \subset B_S(\delta, u) \subset \bar{B}_S(\delta,
% u)$. Further, it follows that $\bbar{C}(\delta) \leq C \leq
% \bar{C}(\delta)$. When delta is ``small'', $\bbar{C}(\delta) \approx C$, which
% implies that $c(\delta)$ is small.

With these definitions, we can now state our result. 
\begin{theorem} \label{thm:mnl-approx} 

  Let $\Mopt_C$ denote the optimal assortment of size at most $C$ and
  $\Mhatopt_C$ denote the estimate of the optimal assortment produced by {\sc
    GreedyOPT} when run with inputs $S \geq 0$ and $b \geq \bar{C}(2\delta_C) +
  1$. Then, we must have
  \begin{equation*}
    \frac{R(\Mopt_C) - R(\Mhatopt_C)}{R(\Mopt_C)} \leq f(w, \beps_{\max}),
  \end{equation*}
  where $w$ denotes the vector of weights $(w_1, w_2, \dotsc, w_N)$ and
  \begin{equation*}
    f(w, \beps_{\max}) \defas \frac{W_C^{\max}}{w(\Mopt_C)} \eta(\beps_{\max})
 \end{equation*}
 with $\eta(\beps_{\max}) \defas 4 C \beps_{\max}/(1 - \beps_{\max})$.

\end{theorem}

It is easy to see that the algorithm calls the revenue function $O(N^2 C^2 \bar{C}(2
\delta_C))$ times. Note that as $\beps_{\max} \to 0$, $\eta(\beps_{\max})$ and
hence $f(w, \beps_{\max})$ go to zero. In addition, it follows from our
definitions that as $\beps_{\max} \to 0$, $\bar{C}(2 \delta_C) \to
C$. Consequently, taking the error in revenues $\beps_{\max} = 0$ yields in
Theorem~\ref{thm:mnl-approx} yields the result of Theorem~\ref{thm:mnl-exact} as
the special result. Therefore, we only prove Theorem~\ref{thm:mnl-approx} in the
next section.

\section{Proofs of the main results} \label{sec:greedyopt_proofs} 

In this section we prove Theorem~\ref{thm:mnl-approx}; specifically, we
establish that the revenues of the optimal assortment and the estimate of the
optimal assortment produced by {\sc GreedyOPT} are ``close''. In order to
establish this result, for the rest of the section, fix a capacity $C$. Let
$\Mopt$ and $\Mhatopt$ respectively denote the optimal assortment and the
estimate of the optimal assortment produced by {\sc GreedyOPT}. Then, our goal
is to show that $R(\Mopt)$ and $R(\Mhatopt)$ are ``close'' to each other. We
assume that the underlying choice model is the MNL model with parameters $w_1,
w_2, \dotsc, w_N$. Recall that for any assortment $\Mscr$,
\begin{equation*}
  R(\Mscr) = \frac{\sum\limits_{i \in \Mscr} p_i w_i}{1 + \sum\limits_{i \in \Mscr} w_i},
\end{equation*}
where $p_i$ is the price of product $i$. The term in the denominator makes
comparison of the revenues of two different assortment difficult. Therefore,
instead of dealing with the revenues of the assortment directly, we consider the
following transformation of the revenues of assortments: for any assortment
$\Mscr$ and number $u \in \R$,
\begin{align*}
  R(\Mscr) - u  = \frac{\sum\limits_{i \in \Mscr} p_i w_i}{1 + \sum\limits_{i \in \Mscr} w_i}
  - u  &= \frac{\left( \sum\limits_{i \in \Mscr} (p_i - u) w_i \right) - u}{1 + \sum\limits_{i
      \in \Mscr} w_i} \\
  &= \frac{H_{\Mscr}(u) - u}{w(\Mscr)},
\end{align*}
where $H_\Mscr \colon \R \to \R$ is a function defined as $H_\Mscr(u) = \sum_{i
  \in \Mscr} (p_i - u) w_i$ and $w(\Mscr) \defas 1 + \sum_{i \in \Mscr} w_i$. We
can now write
\begin{equation}
  \label{eq:greedyOPT1}
    H_\Mscr(u) = u  + w(\Mscr) (R(\Mscr)  - u).
\end{equation}
It is clear that $H_\Mscr(\cdot)$ is directly related to the revenue
$R(\Mscr)$. Moreover, as will become apparent soon, it is easier to compare the
transformations $H_{\Mscr_1}(\cdot)$ and $H_{\Mscr_2}(\cdot)$ of two assortments
$\Mscr_1$ and $\Mscr_2$ than their revenues $R(\Mscr_1)$ and
$R(\Mscr_2)$. Specifically, we can establish the properties stated in the
following proposition. 

\begin{proposition} \label{prop:fact-RH-approx}
  For any two assortments $\Mscr_1$ and $\Mscr_2$, let $H_1(\cdot)$ and
  $H_2(\cdot)$ respectively denote the functions $H_{\Mscr_1}(\cdot)$ and
  $H_{\Mscr_2}(\cdot)$. Further, let $u_1$ and $u_2$ denote the revenues
  $R(\Mscr_1)$ and $R(\Mscr_2)$ respectively. We then have
  \begin{enumerate}
  \item $H_1(u_2) \geq H_2(u_2) \iff R(\Mscr_1) \geq R(\Mscr_2)$.  
  \item $H_1(u_2) \geq  (1 + \delta(\Mscr_1))H_2(u_2) \implies \Rtilde(\Mscr_1)
    \geq \Rtilde(\Mscr_2)$,
  \end{enumerate}
  where $\delta(\Mscr_1) \defas \beps(\Mscr_1) w(\Mscr_1)/(1 -
  \beps(\Mscr_1))$.
\end{proposition}
\begin{proof}
  We prove each of the properties in turn. First note that for any assortment
  $\Mscr$ with revenue $R(\Mscr) = u$, it immediately follows from our
  definitions that $H_\Mscr(u) = u + w(\Mscr) (R(\Mscr) - u) = u$. The first
  property now follows from a straightforward expansion of the terms
  involved:
  \begin{align*}
    H_1(u_2) \geq H_2(u_2) & \iff u_2 + w(\Mscr_1) (u_1 -  u_2) \geq u_2 \\
    &\iff u_1 \geq u_2 \\
    & \iff R(\Mscr_1) \geq R(\Mscr_2),
  \end{align*}
  where the second equivalence follows from the fact that $w(\Mscr_1) > 0$.  The
  second property can also be obtained through a similar straightforward
  expansion of the terms. In particular,
  \begin{align}
    H_1(u_2) \geq (1 + \beps(\Mscr_1)) H_2(u_2) &\iff u_2 + w(\Mscr_1)
    (u_1 - u_2) \geq (1 + \delta(\Mscr_1))u_2 \nonumber\\
    &\iff u_1 \geq \left( 1 + \frac{\delta(\Mscr_1)}{w(\Mscr_1)}
    \right) u_2 \nonumber\\
    &\iff u_1 \geq \left( 1 + \frac{\beps(\Mscr_1)}{1 -
        \beps(\Mscr_1)} \right) u_2 \nonumber\\
    &\iff (1 - \beps(\Mscr_1)) u_1 \geq u_2, \label{eq:appdecisions_1}
  \end{align}
  where the second equivalence follows from the definition of $\delta(\Mscr_1)$.
  Moreover, it follows from our definitions that $\Rtilde(\Mscr_1) \geq (1 -
  \beps(\Mscr_1)) u_1$ and $u_2 \geq \Rtilde(\Mscr_2)$. We now conclude from
  \eqref{eq:appdecisions_1} that
  \begin{align*}
    \Rtilde(\Mscr_1) \geq (1 - \beps(\Mscr_1)) u_1 \geq u_2 \geq
    \Rtilde(\Mscr_2).
  \end{align*}
  The result of the proposition now follows. 
\end{proof}

The above proposition establishes that if the transformation $H_{\Mscr}(\cdot)$
of one assortment is ``sufficiently'' larger than the other, then it follows that
the revenues of one assortment should be larger than the revenues of the
other. Therefore, instead of keeping track of the revenues of the assortments in
our algorithm, we keep track of their respective transformations
$H_{\Mscr}(\cdot)$. 

Next, we establish a loop-invariance property that arises due to greedy
additions and exchanges in our algorithms. We make use of this property to prove
our theorems. In order to state the proposition, we introduce the following
notation:
\begin{equation*}
  \delta_C \defas \max_{\Mscr \colon \abs{\Mscr}} \delta(\Mscr) = \max_{\Mscr
    \colon \abs{\Mscr}} w(\Mscr) \frac{\beps(\Mscr)}{1 - \beps(\Mscr)}.
\end{equation*}
We then have
\begin{proposition} \label{prop:mnl-AX-approx} 

  Consider an iteration $t$ of the while loop of the {\sc GreedyADD-EXCHANGE}
  algorithm. Let $\Mscr_t$ and $\Mscr_{t+1}$ denote the estimates of the optimal
  assortments at the beginning and the end of iteration $t$. Let $\Nscr_t$
  denote the universe of products at the beginning of iteration $t$. Then, 
  \begin{enumerate}
  \item if a greedy exchange takes place i.e., $\Mscr_{t+1} = \left( \Mscr_t
      \setminus \set{i^*} \right) \cup \set{j^*}$, then for $u =
    R(\Mscr_{t+1})$, we must have
    \begin{align*}
      h_{i^*}(u) &\leq h_i(u) + \delta_C u, &&\text{ for
        all
      } i \in \Mscr_t\\
      h_{j^*}(u) &\geq h_j(u) - \delta_C u, &&\text{ for
        all } j \in \Nscr_t \setminus \Mscr_t;
    \end{align*}
  \item if an addition takes place i.e., $\Mscr_{t+1} = \Mscr_t \cup \set{j^*}$,
    then for $u = R(\Mscr_{t+1})$ we must have
    \begin{align*}
      h_{j^*}(u) &\geq h_j(u) - \delta_C u, \quad\text{for
        all } j \in \Nscr_t \setminus \Mscr_t.
    \end{align*}
  \end{enumerate}
\end{proposition}
\begin{proof}
  We prove this proposition by contradiction. First consider the case when
  exchange happens i.e., $\Mscr_{t+1} = \left( \Mscr_t \setminus \set{i^*}
  \right) \cup \set{j^*}$. Note that for any assortment $\Mscr = \left(\Mscr_t
    \setminus \set{i} \right) \cup \set{j}$ with $i \in \Mscr_t$ and $j \in
  \Nscr_t \setminus \Mscr_t$, letting $u$ denote $R(\Mscr_{t+1})$, we can write
  \begin{equation}
    \label{eq:appdecisions_4}
    H_{\Mscr}(u) - H_{\Mscr_{t+1}}(u) = h_j(u) -
    h_{j^*}(u) + h_{i^*}(u) - h_{i}(u). 
  \end{equation}
  Now, if the hypothesis of the proposition pertaining to exchange is false,
  then at least one of the following should be true: either (1) there exists a
  product $i \in \Mscr_t$ and $i \neq i^*$ such that $h_{i^*}(u) > h_i(u) +
  \delta_C u$, or (2) there exists a product $j \in \Nscr_t \setminus \Mscr_t$
  and $j \neq j^*$ such that $h_{j^*}(u) < h_j(u) + \delta_C u$.  In the first
  case when $h_{i^*}(u) > h_i(u) + \delta_C u$, by taking $j = j^*$, we can
  write from~\eqref{eq:appdecisions_4} that $H_{\Mscr}(u) - H_{\Mscr_{t+1}}(u) >
  \delta_C u$. Similarly, in the second case when $h_{j^*}(u) < h_j(u) +
  \delta_C u$, by taking $i = i^*$, we can write from~\eqref{eq:appdecisions_4}
  that $H_{\Mscr}(u) - H_{\Mscr_{t+1}}(u) > \delta_C u$. Therefore, in both the
  cases, we have exhibited an assortment $\Mscr$ distinct from $\Mscr_{t+1}$
  that can be obtained from $\Mscr_t$ through an exchange and has the property
  that $H_{\Mscr}(u) - H_{\Mscr_{t+1}}(u) > \delta_C u$. We can now write
  \begin{subequations} \label{eq:appdecisions_subeqs}
    \begin{align}
    &&H_{\Mscr}(u) &> H_{\Mscr_{t+1}}(u) + \delta_C u &&\\
    \implies &&H_{\Mscr}(u) &> H_{\Mscr_{t+1}}(u) + \delta(\Mscr) u &&
    \text{since } \delta_C \geq \delta(\Mscr) \text{ by definition }\\
    \implies &&H_{\Mscr}(u) &> (1 + \delta(\Mscr)) H_{\Mscr_{t+1}}(u) &&
    \text{since } H_{\Mscr_{t+1}}(u) = u \text{ by definition}\\
    \implies &&\Rtilde(\Mscr) &> \Rtilde(\Mscr_{t+1}) && \text{by
      Proposition~\ref{prop:fact-RH-approx}}.
  \end{align}
  \end{subequations}
  This clearly contradicts the fact that $\Mscr_{t+1}$ is chosen greedily.

  The case when addition happens can be proved in the exact similar
  way. Particularly, suppose there exists a product $j \in \Nscr_t \setminus
  \Mscr_t$ and $j \neq j^*$ such that $h_{j^*}(u) < h_j(u) - \delta_C u$, where
  $u = R(\Mscr_{t+1})$ with $\Mscr_{t+1} = \Mscr_t \cup \set{j}$. Letting
  $\Mscr$ denote the set $\Mscr_t \cup \set{j}$, we can then write
  \begin{equation*}
    H_{\Mscr}(u) - H_{\Mscr_{t+1}}(u)  =  h_j(u) - h_{j^*}(u) > \delta_C u.
  \end{equation*}
  This implies -- following the sequence of arguments
  in~\eqref{eq:appdecisions_subeqs} -- that $\Rtilde(\Mscr) > \Rtilde(\Mscr_t)$,
  contradicting the fact that $\Mscr_{t+1}$ is chosen greedily.  
  
  The result of the proposition now follows.
\end{proof}

The above proposition establishes a key loop-invariance property that results
from greedy additions and exchanges. Specifically, let $u$ denote the revenue of
the estimate of the optimal assortment obtained at the end of an iteration of
the while loop in {\sc GreedyADD-EXCHANGE}. Then, the proposition establishes
that whenever a product $j^*$ is introduced (either through addition or an
exchange-in) greedily, it must be that $h_{j^*}(u)$ is ``close'' to the maximum
$h_j(u)$ of all products $j$ that have been considered for an addition or
exchange-in. Similarly, the product $i^*$ that is greedily exchanged-out must be
such that $h_{i^*}(u)$ is ``close'' to the minimum $h_i(u)$ of all products $i$
that have been considered for an exchange-out.

Using the propositions above, we can establish a key property of the
subroutine {\sc GreedyADD-EXCHANGE}. For that, we need the following
notation. For any $u$, define
\begin{equation*}
  B_S(u) \defas \arg\max_{\Mscr \colon \abs{\Mscr} \leq S} H_{\Mscr}(u)  =
  \arg\max_{\Mscr \colon \abs{\Mscr} \leq S} \sum_{i \in \Mscr} h_i(u).
\end{equation*}
It is easy to see from the above definition that $B_S(u)$ consists of the top at
most $C$ products according to $h_i(u)$ such that $h_i(u) > 0$. Since
$h_i(\cdot)$ is monotonically decreasing, it is easy to see that
\begin{equation}
  \label{eq:appdecision_11}
  \abs{B_S(u_1)} \geq \abs{B_S(u_2)}, \quad \text{whenever } u_1 \leq u_2.
\end{equation}
Under appropriate technical assumptions,~\cite{Rus10} showed that for any $1
\leq S \leq N$, the optimal assortment of size at most $S$ under the MNL model
is one of the assortments in the collection $\Bscr_S \defas \set{B_{S}(u) \colon
  u \in \R}$. In fact the authors show that if $u_S$ denotes the optimal
revenue, then $B_S(u_S)$ is the optimal assortment. An immediate consequence of
this result and~\eqref{eq:appdecision_11} is that for any $u \leq u_S$
\begin{equation}
  \label{eq:appdecision_bs}
  S \geq \abs{B_S(u)} \geq \abs{\Mopt_S}.
\end{equation}
It has been established by~\cite{Rus10} that there can be at most $O(NC)$
distinct assortments in the collection $\Bscr_S$ allowing one to find the
optimal assortment by restricting one's search to $O(NC)$ assortments. The
following lemma shows that the assortment found by the subroutine {\sc
  GreedyADD-EXCHANGE} is ``close'' to one of the assortments in
$\Bscr_{S}$. Before we describe the lemma, we need the following notation. 
For any $\delta > 0$ and $u \in \R$, let 
\begin{equation*}
  i_S(u)  \defas \min_{i \in B_S(u)} h_i(u).
\end{equation*}
Moreover, let 
\begin{equation*}
  \bar{B}_S(\delta, u) \defas B_S(u) \cup \set{j \in \Nscr\setminus B_S(u) \colon
    h_{i_S(u)}(u) - h_j(u)\leq \delta u},
\end{equation*}
Also, let
\begin{equation*}
  \bar{C}(\delta) \defas \max_{u \in \Rp} \abs{\bar{B}_S(\delta, u)},
\end{equation*}
We then have
\begin{lemma} \label{lem:mnl-approx}
  Suppose {\sc GreedyADD-EXCHANGE} is run with some input assortment
  $\Mscr$ and $b \geq \bar{C}(\delta_C) + 2$, where $C \geq S+1$. Further,
  suppose that $\abs{\Mopt_{S+1}} = S+1$. Then, there exists an iteration $t^*$
  of the while loop such that if $\Mscr^*$ denotes the assortment $\Mscr_{t^* +
    1}$ and $u^*$ denotes $R(\Mscr^*)$, then
  \begin{equation*}
    H_{B(u^*)}(u^*) - H_{\Mscr^*}(u^*) \leq 2 \Ctilde_{u^*} \delta_C u^*,
  \end{equation*}
  where $B(u^*)$ denotes the assortment $B_{S+1}(u^*)$ and $\Ctilde^*$ is a
  constant denoting $1 +  \abs{B(u^*) \setminus \Mscr^*}$.
\end{lemma}
We defer the proof of Lemma~\ref{lem:mnl-approx} to the end of the
section. We now present the proof of Theorem~\ref{thm:mnl-approx}.

\subsection{Proof of Theorem~\ref{thm:mnl-approx}}
Let $\Mopt_C$ denote the true optimal assortment, and $\Mhatopt_C$
denote the estimate of the optimal assortment produced by {\sc
  GreedyOPT}. Furthermore, let $C^* \leq C$ denote the size of
$\Mopt_C$. It follows from Lemma~\ref{lem:mnl-approx} that in the
$C^*$th invocation of the subroutine {\sc GreedyADD-EXCHANGE}, there
exists an assortment $\Mscr^*$ such that $\Rtilde(\Mhatopt_C) >
\Rtilde(\Mscr^*)$ and $\Mscr^*$ is such that
\begin{equation*}
  H_{B(u^*)}(u^*) - H_{\Mscr^*}(u^*) \leq 2 \Ctilde_{u^*} \delta_C u^*,
\end{equation*}
where $\Ctilde^*$ denotes $\abs{B(u^*) \setminus \Mscr^*} + 1$ and
$B(u^*)$ denotes the set $B_{C^*}(u^*)$. It follows by the definition
of $B(u^*)$ that $H_{B(u^*)}(u^*) \geq H_{\Mopt_C}(u^*)$. Thus, we can
write
\begin{equation}
  \label{eq:mnl-approx-thm2}
  H_{\Mopt_C}(u^*) - H_{\Mscr^*}(u^*) \leq 2 \Ctilde_{u^*} \delta_C
  u^* \leq 2 C \delta_C u^*.
\end{equation}
Let $u_C$ denote $R(\Mopt_C)$. Then, it follows by definition that
$H_{\Mopt_C}(u_C) = u_C$. Thus,
\begin{align*}
  H_{\Mopt_C}(u_C) - H_{\Mopt_C}(u^*) = \sum_{j \in \Mopt_C} w_j(u^* -
  u_C) = (u^* - u_C) (w(\Mopt_C) - 1).
\end{align*}
Since $H_{\Mopt_C}(u_C) = u_C$, we can write
\begin{equation}
  \label{eq:mnl-approx-thm3}
  H_{\Mopt_C}(u^*) = u_C + (u_C - u^*)(w(\Mopt_C) - 1).
\end{equation}
Since $H_{\Mscr^*}(u^*) = u^*$, it now follows from
\eqref{eq:mnl-approx-thm2} and \eqref{eq:mnl-approx-thm3} that
\begin{align}
  &&(u_C - u^*)(w(\Mopt_C) - 1) + u_C - u^* &\leq 2 C \delta_C u^* \nonumber\\
\implies && (u_C - u^*) w(\Mopt_C) &\leq 2 C \delta_C u^* \nonumber\\
\implies &&u_C &\leq (1+\tilde{\beps}) u^*, \label{eq:mnl-approx-thm4}
\end{align}
where $\tilde{\beps} \defas 2 C \delta_C /w(\Mopt_C)$. Now since
$\Rtilde(\Mhatopt_C) > \Rtilde(\Mscr^*)$, it follows that
\begin{equation*}
(1 -\beps(\Mscr^*)) u^* \leq \Rtilde(\Mscr^*) < \Rtilde(\Mhatopt) \leq \uhat_C,
\end{equation*}
where $\uhat_C$ denotes $R(\Mhatopt_C)$. It now follows
from~\eqref{eq:mnl-approx-thm4} that
\begin{equation*}
  u_C \leq (1 + \tilde{\beps}) u^* \leq \frac{1 + \tilde{\beps}}{1 - \beps(\Mscr^*)} \uhat_C.
\end{equation*}

Now since 
\begin{equation*}
  \delta_C = \max_{\Mscr \colon \abs{\Mscr} \leq C}
  \frac{\beps(\Mscr)}{ 1- \beps(\Mscr)} w(\Mscr),
\end{equation*}
by letting $\beps_{\max} = \max_{\Mscr \colon \abs{\Mscr} \leq C} \beps(\Mscr)$ and
$W_C^{\max} = \max_{\Mscr \colon \abs{\Mscr} \leq C} w(\Mscr)$, we
have
\begin{equation*}
  \delta_C \leq  \frac{\beps_{\max}}{1 - \beps_{\max}} W_C^{\max}.
\end{equation*}
Thus,
\begin{align*}
  \tilde{\beps} = \frac{2C}{w(\Mopt_C)} \delta_C \leq \frac{2
    C}{w(\Mopt_C)} \frac{\beps_{\max}}{1  - \beps_{\max}} W_C^{\max}
  \defas f(w, \beps_{\max})/2.
\end{align*}
With these definitions, it is easy to see that $\beps(\Mscr^*) \leq
\beps_{\max} \leq f(w, \beps_{\max})/2$. It now follows that
\begin{equation*}
  \frac{u_C - \uhat_C}{u_C} \leq 1 - \frac{1 - \beps(\Mscr^*)}{1  +
    \tilde{\beps}} \leq \frac{\tilde{\beps} + \beps(\Mscr^*)}{1 + \tilde{\beps}}
  \leq \beps(\Mscr^*) + \tilde{\beps} \leq f(w, \beps_{\max}).
\end{equation*}
This establishes the result of the theorem.

\subsection{Proof of Lemma~\ref{lem:mnl-approx}}
Suppose the while loop in the subroutine terminates at the end of
iteration $T$. Then, it follows from the description of the subroutine
that at least one of the following conditions holds at the end of
iteration $T$:
\begin{enumerate}
\item The set of products $\Nscr_{T+1} \setminus \Mscr_{T+1}$ available for
  additions or exchanges is empty.
\item No further additions or exchanges can increase the revenues.
\end{enumerate}
Our goal is to prove the existence of an iteration $t^* \leq T$ such that
\begin{equation*}
  H_{B(u^*)}(u^*) - H_{\Mscr^*}(u^*) \leq 2 \Ctilde_{u^*} \delta_C u^*,
\end{equation*}
where $\Mscr^*$ denotes the assortment $\Mscr_{t^* + 1}$ and $u^*$ denotes
$R(\Mscr^*)$. We prove this by considering two cases corresponding to each of
the two ways in which the subroutine terminates. Note that in order to simplify
the notation, we have dropped the subscript from the notation of
$B_{S+1}(\cdot)$.

\noindent{\bf Case 1: Subroutine terminates with $\Nscr_{T+1} = \Mscr_{T+1}$.}
We first consider the case when the subroutine terminates when the set of
products $\Nscr_{T+1} \setminus \Mscr_{T+1}$ becomes empty. In this case, we
prove the existence of an iteration $t^* \leq T$ that satisfies the condition
stated in the hypothesis of the lemma. In fact, we prove something stronger; we
shall show that the iteration $t^* \leq T^*$, where $T^* \leq T$ is the first
iteration such that $\Nscr_{T^*} \subset \Nscr$ (recall that $\Nscr_1 =
\Nscr$). We prove this result by contradiction. In particular, suppose that
after every iteration $t \leq T^*$ of the while loop, we have
\begin{equation}
  \label{eq:appdecision_5}
  H_{B(u)}(u) - H_{\Mscr_{t+1}}(u) > 2 \Ctilde_{u} \delta_C u, 
\end{equation}
where $u$ denotes the revenue $R(\Mscr_{t+1})$ and $\Ctilde_u$ denotes the
constant $1 + \abs{B(u) \setminus \Mscr_{t+1}}$. Note that a product $i$ would be
removed from the universe $\Nscr_t$ at the end of some iteration $t$ only if it
has been exchanged-out $b$ times. Since $b \geq \bar{C}(\delta_C)$, it is easy
to see that we arrive at a contradiction if we show that as
long~\eqref{eq:appdecision_5} is satisfied at the end of each iteration, each
product $i$ can be exchanged-out at most $\bar{C}(\delta_C) + 2$ times.

In order to bound the number of times a product can be exchanged-out, we
establish a special property that should be satisfied whenever an exchange
happens. Specifically, suppose an exchange happens during iteration $t$ i.e.,
$\Mscr_{t+1} = \left( \Mscr_t \setminus \set{i^*} \right) \cup \set{j^*}$. In
addition, let $u$ denote the revenue $R(\Mscr_{t+1})$, and let product $k^* \in
\Nscr_t \setminus \Mscr_t$ denote the product such that $h_{k^*}(u) \geq h_k(u)$
for all products $k \in \Nscr_t \setminus \Mscr_t$. Then, we claim that
\begin{subequations} \label{eq:appdecision_6}
  \begin{align}
    h_{j^*}(u) &\geq h_{k^*}(u) - \delta_C u \label{eq:appdecision_6a} \\
    h_{i^*}(u) &\leq h_{k^*}(u) - \delta_C u \label{eq:appdecision_6b}.
  \end{align}
\end{subequations}
We prove this claim as follows. Since $k^* \in \Nscr_t \setminus \Mscr_t$,
\eqref{eq:appdecision_6a} follows directly from
Proposition~\ref{prop:mnl-AX-approx}. We now argue that $h_{i^*}(u) \leq
h_{k^*}(u) - \delta_C u$. For that, we first note that
\begin{equation} \label{eq:appdecision_8} 
  h_{i^*}(u) - h_{j^*}(u) \leq 2 \delta_C u.
\end{equation}
To see why, note that since an exchange has happened, it must be that
$\Rtilde(\Mscr_t) \leq \Rtilde(\Mscr_{t+1})$. This implies by
Proposition~\ref{prop:fact-RH-approx} that $H_{\Mscr_t}(u) \leq (1 +
\delta(\Mscr_t)) H_{\Mscr_{t+1}}(u)$. Since $\delta(\Mscr_1) \leq \delta_C$ and
$H_{\Mscr_{t+1}}(u) = u$ by definition, we can write
\begin{align*}
  H_{\Mscr_t}(u) \leq (1 + \delta(\Mscr_t)) H_{\Mscr_{t+1}}(u)
  &\implies H_{\Mscr_t}(u) - H_{\Mscr_{t+1}}(u) \leq \delta_C u \\
  &\implies h_{i^*}(u) - h_{j^*}(u) \leq \delta_C u < 2 \delta_C u.
\end{align*}
Now, consider
\begin{align*}
  H_{B(u)}(u) - H_{\Mscr_{t+1}}(u) &= H_{B(u)}(u) - H_{\Mscr_t}(u) +
  H_{\Mscr_t}(u) - H_{\Mscr_{t+1}}(u) \\
  &= \sum_{j \in B(u) \setminus \Mscr_t} h_j(u) - \sum_{i \in \Mscr_t \setminus
    B(u)} h_i(u) + \left(h_{i^*}(u) - h_{j^*}(u) \right).
\end{align*}
We now collect terms in the above expression as follows. Let $\Mscr_1$ denote
the set $\Mscr_t \setminus B(u)$. Further, partition the set $B(u) \setminus
\Mscr_t$ into $Mscr_2 \cup \Mscr_3$ such that $\Mscr_2 \cap \Mscr_3 = \emptyset$
and $\abs{\Mscr_2} = \abs{\Mscr_1}$; note that such a partitioning is possible
because $\abs{B(u)} = S+1$ (which follows from~\eqref{eq:appdecision_bs} and the
hypothesis that $\abs{\Mopt_{S+1} = S+1}$) and $\abs{\Mscr_t} \leq S+1$. Also
note that $\Mscr_3 \neq \emptyset$ if and only if $\abs{\Mscr_t} < S+1$. With
this partitioning, we can now write
\begin{equation*}
  H_{B(u)}(u) - H_{\Mscr_{t+1}}(u) = \sum_{i \in \Mscr_1, j \in
    \Mscr_2} \left(  h_j(u) - h_i(u) \right) + \sum_{j \in \Mscr_3}
  h_j(u) + \left( h_{i^*}(u) - h_{j^*}(u) \right).
\end{equation*}
We now claim that at least on of the following must be true: either (1) there
exists a pair of products $i \in \Mscr_1$ and $j \in \Mscr_2$ such that $h_j(u)
- h_i(u) > 2 \delta_C u$, or (2) if $\Mscr_3 \neq \emptyset$, then there exists
a product $k \in \Mscr_3$ such that $h_3(u) > 2 \delta_C u$. Otherwise, it is
easy to see from~\eqref{eq:appdecision_8} that $H_{B(u)}(u) - H_{\Mscr_{t+1}}(u)
\leq 2 \Ctilde_u \delta_C u$, where $\Ctilde_u = \abs{B(u) \setminus
  \Mscr_{t+1}} + 1$, contradicting \eqref{eq:appdecision_5}. We now consider
each of the cases in turn.

First suppose that $h_j(u) - h_i(u) > 2 \delta_C u$ for some $i \in \Mscr_1$ and
$j \in \Mscr_2$. It follows from Proposition~\ref{prop:mnl-AX-approx} that
$h_{i^*}(u) \leq h_{i}(u) + \delta_C u$. Thus, we can write
\begin{equation*}
  h_{i^*}(u) \leq h_{i}(u) + \delta_C u < h_j(u) - 2 \delta_C u +
  \delta_C u \leq h_{k^*}(u) - \delta_C u,
\end{equation*}
where the last inequality follows from the definition of $k^*$ and the fact that
$j \in \Mscr_2 \subset \Nscr \setminus \Mscr_t$. Thus, for this case, we have
established~\eqref{eq:appdecision_6b}.

Now suppose that $\Mscr_3 \neq \emptyset$ and $h_k(u) > 2 \delta_C u$ for some
$k \in \Mscr_3$. As noted above, in this case, we should have $\abs{\Mscr_{t+1}}
< S+1$. This means that an exchange has happened instead of addition, which in
turn implies that $\Rtilde(\Mtilde) \leq \Rtilde(\Mscr_{t+1})$, where $\Mtilde$
denotes the set $\Mscr_t \cup \set{k}$. Thus, by
Proposition~\ref{prop:fact-RH-approx}, we should have
\begin{align*}
  &&H_{\Mtilde}(u) \leq (1 + \delta(\Mtilde)) H_{\Mscr_{t+1}}(u) \\
  \implies &&H_{\Mtilde}(u) - H_{\Mscr_{t+1}}(u) \leq \delta(\Mtilde)
  H_{\Mscr_{t+1}}(u) &&\\
  \implies &&h_k(u) + h_{i^*}(u) - h_{j^*}(u) \leq \delta_C u && \text{as }
  H_{\Mscr_{t+1}}(u) = u, \delta(\Mtilde) \leq
  \delta_C \\
  \implies &&h_{i^*}(u) \leq h_{j^*}(u) - h_{k}(u) + \delta_C u &&\\
  \implies &&h_{i^*}(u) \leq h_{j^*}(u) - 2 \delta_C u + \delta_C u
  && \text{ since } h_{k}(u) > 2 \delta_C u\\
  \implies && h_{i^*}(u) \leq h_{k^*}(u) - \delta_C u && \text{ since }
  h_{j^*}(u) \leq h_{k^*}(u).
\end{align*}
We have thus established that $h_{i^*}(u) \leq h_{k^*}(u) - \delta_C u$ for both
the cases.

We now use~\eqref{eq:appdecision_6} to bound the number of exchange-outs that
can happen for each product. Specifically, as mentioned above, we arrive at a
contradiction by showing that each product can be exchanged-out at most
$\bar{C}(\delta_C) + 2$ times. For that, for any iteration $t \leq T^*$, let
$k_t$ denote the product such that $k_t \in \Nscr_t \setminus \Mscr_t$ and
$h_{k_t}(u_{t+1}) \geq h_j(u_{t+1})$ for all products $j \in \Nscr_t \setminus
Mscr_t$ and $u_{t+1} = R(\Mscr_{t+1})$. Now define the function
\begin{equation*}
  g(u) =
  \begin{cases}
    h_{k_t}(u) - \delta_C u & \text{ for } u_t < u \leq u_{t+1}, t
    \leq T^*,\\
    h_{k_1}(u_1) - \delta_C u_1 & \text{ for } u = u_1.
  \end{cases}
\end{equation*}
Note that for the above definition to be meaningful, for any $t \leq T^*$, we
need to show that $u_t \leq u_{t+1}$. This should be true because
by~\eqref{eq:appdecision_6}, it follows that for $u = R(\Mscr_{t+1})$, we have
$h_{i^*}(u) \leq h_{j^*}(u)$; this in turn implies that $H_{\Mscr_t}(u) \leq
H_{\Mscr_{t+1}}(u)$, which implies by Proposition~\ref{prop:fact-RH-approx} that
$u_t = R(\Mscr_t) \leq R(\Mscr_{t+1}) = u_{t+1}$. It is easy to see that the
function $g(\cdot)$ is piecewise linear. However, note that it may not be
continuous.
  
Now fix a product $i$, and for this product we argue that it can be exchanged at
most $\bar{C}(\delta_C)$ times. For that let $t_1$ be an iteration in which $i$
is exchanged-out and $t_2$ be the first iteration after $t_1$ when $i$ is
exchanged-in. Let $u_1$, $u_2$ denote $R(\Mscr_{t_1 + 1})$ and $R(\Mscr_{t_2 +
  1})$ respectively. Furthermore, let $k_1$ and $k_2$ respectively denote the
products $k_{t_1}$ and $k_{t_2}$. It now follows from~\eqref{eq:appdecision_6}
that
\begin{align*}
  h_i(u_1) &\leq h_{k_1}(u_1) - \delta_C u_1 = g(u_1) \\
  h_i(u_2) &\geq h_{k_2}(u_2) - \delta_C u_2 = g(u_2).
\end{align*}
This implies that the line $h_i(\cdot)$ is below $g(\cdot)$ at $u_1$ and above
$g(\cdot)$ at $u_2$. We now argue that $h_i(\cdot)$ intersects $g(\cdot)$ at
some $u_1 \leq u \leq u_2$ i.e., $h_i(u) = g(u)$. If $g(\cdot)$ were continuous,
this assertion would immediately follow from the intermediate value
theorem. However, the way we have defined $g(\cdot)$, it may be discontinuous at
some $u_t$ with $t_1 < t \leq t_2$. Now the only way $h_i(\cdot)$ and $g(\cdot)$
do not intersect is if for some $t_1 < t \leq t_2$,
\begin{equation*}
  g(u_t^{-}) < h_i(u_t) < g(u_t^{+}) \quad \text{and} \quad h_i(u)
  > g(u) \text{ for } u_t \leq u \leq u_2.
\end{equation*}
We argue that this cannot happen. For that consider iteration $t$. By definition
$i \notin \Mscr_t$. Since $\Nscr_t = \Nscr$, it follows by our definition that
$h_{k_t}(u_{t+1}) \geq h_i(u_{t+1})$, which in turn implies that $g(u_{t+1})
\geq h_i(u_{t+1})$ resulting in a contradiction. Thus, $h_i(\cdot)$ intersects
$g(\cdot)$ from below at some $u$ such that $u_1 \leq u \leq u_2$.

Hence, we can correspond each exchange-out with an intersection point
corresponding to $h_i(\cdot)$ intersecting $g(\cdot)$ from below. This implies
that the total number of exchage-outs can be bounded above by one plus the
number of times $h_i(\cdot)$ intersects $g(\cdot)$ from below beyond $u_i$,
where $u_i$ is the revenue of the assortment $\Mscr_t$ immediately after $i$ is
added to it (either through an exchange-in or addition). Note that $h_i(\cdot)$
intersects $g(\cdot)$ at $u \geq u_i$ if and only if $w_i \leq w_{k(u)}$ and
$h_{k(u)}(u_i) \geq h_i(u_i)$, where $k(u)$ is the product such that $k(u) =
k_t$, where $u_t < u \leq u_{t+1}$. Thus, the number of intersection points can
be bounded above by the number of products $k$ such that $h_k(u_i) \geq
h_i(u_i)$. We now argue that $i \in \bar{B}_{S+1}(\delta_C, u_i)$. If this is
true, then it implies that there can be at most $\abs{\bar{B}_{S+1}(\delta_C,
  u_i)} \leq \bar{C}(\delta_C)$ intersection points, which immediately implies
that there can be at most $1 + \bar{C}(\delta_C)$ exchange-outs.

The only thing we are left with is to argue that $i \in \bar{B}_{S+1}(\delta_C,
u_i)$. To see this, let $\Mtilde$ be the assortment obtained after $i$ is added
or exhanged-in for the first time. Then, according to our definition, we have
that $u_i = R(\Mtilde)$. Further, since $H_{B(u_i)}(u_i) - H_{\Mtilde}(u_i) >
0$, there exists a product $k \in B(u_i) \setminus \Mtilde$. It now follows by
Proposition~\ref{prop:mnl-AX-approx} that
\begin{equation*}
  h_i(u_i) \geq h_k(u_i) - \delta_C u_i \geq h_{i_{S+1}(u_i)}  -
  \delta_C u_i,
\end{equation*}
where $i_{S+1}(u_i)$ is as defined above i.e., $i_{S+1}(u_i) \defas \arg\min_{j
  \in B(u_i)} h_j(u_i)$. It now follows by the definition of
$\bar{B}_{S+1}(\delta_C, u_i)$ that $i \in \bar{B}_{S+1}(\delta_C, u_i)$.

\noindent{\bf Case 2: Subroutine terminates because no further additions or
  exchanges increase revenue.} We now consider the case when subroutine
terminates at iteration $T$ because no further additions or exchanges increase
the revenue. Now there are two possibilities: either $\Nscr_t = \Nscr$ for all
$t \leq T$ or not. In the latter case let $T^*$ be the first iteration $t$ when
$\Nscr_t \subset \Nscr$. It then follows from our arguments for the above case
that there exists an iteration $t^* \leq T^*$ that satisfies the properties of
the lemma. Thus, we consider the case when $\Nscr_t = \Nscr$ for all $t \leq
T^*$. Assuming this, we prove the result by contradiction. In particular, suppose
at the end of iteration $T$ we have
\begin{equation} \label{eq:appdecision_10} H_{B(u)}(u) - H_{\Mscr_{T+1}}(u) \geq
  2 \Ctilde_u \delta_C u,
\end{equation}
Now consider
\begin{equation*}
  H_{B(u)}(u) - H_{\Mscr_{T+1}}(u) = \sum_{k \in \Mscr_3} h_j(u) +
  \sum_{i \in \Mscr_1, j \in \Mscr_2} \left(h_{j}(u) - h_i(u) \right),
\end{equation*}
where as above, $\Mscr_1$ denotes the assortment $\Mscr_{T+1} \setminus B(u)$
and the set $B(u) \setminus \Mscr_{T+1}$ is partitioned into $\Mscr_2 \cup
\Mscr_3$ such that $\Mscr_2 \cap \Mscr_3 = \emptyset$ and $\abs{\Mscr_2} =
\abs{\Mscr_1}$; such a partitioning is possible since $\abs{B(u)} = S+1$ (which
follows from~\eqref{eq:appdecision_bs} and the hypothesis that $\abs{\Mopt_{S+1}
  = S+1}$) and $\abs{\Mscr_{T+1}} \leq S+1$. It now follows that one of the
following conditions should hold: either (1) there exists a pair of products $i
\in \Mscr_1$ and $j \in \Mscr_2$ such that $h_j(u) - h_i(u) > 2 \delta_C u$, or
(2) if $\Mscr_3 \neq \emptyset$, then there exists a product $k \in \Mscr_3$
such that $h_3(u) > 2 \delta_C u$. Otherwise, it is easy to see that
$H_{B(u)}(u) - H_{\Mscr_{T+1}}(u) \leq 2 \Ctilde_u \delta_C u$, where $\Ctilde_u
= \abs{B(u) \setminus \Mscr_{T+1}} + 1$, contradicting
\eqref{eq:appdecision_10}. We consider each of the cases in turn.

First, suppose that there exist a pair of products $i \in \Mscr_1$ and $j \in
\Mscr_2$ such that $h_j(u) - h_i(u) > 2 \delta_C u$. Let $\Mtilde$ denote the
assortment $\left(\Mscr_{T+1} \setminus \set{i} \right) \cup \set{j}$. We can
then write
\begin{equation*}
  H_{\Mtilde}(u) - H_{\Mscr_{T+1}}(u) = h_j(u) - h_i(u) > 2
  \delta_C u.
\end{equation*}
Since $H_{\Mscr_{T+1}}(u) = u$ and $\delta_C \geq \delta(\Mtilde)$, it follows
that by Proposition~\ref{prop:fact-RH-approx} that $\Rtilde{\Mtilde} >
\Rtilde{\Mscr_{T+1}}$. This contradicts the assumption that the subroutine
terminates with $\Mscr_{T+1}$ because no further additions or exchanges result
in an increase of revenue.

Next, suppose $\Mscr_3 \neq \emptyset$ and $h_k(u) > 2 \delta_C u$ for some $k
\in \Mscr_3$. Now let $\Mtilde = \Mscr_{T+1} \cup \set{k}$; note that since
$\Mscr_3 \neq \emptyset$, it must be that $\abs{\Mscr_{T+1}} = S$. We can now
write
\begin{equation*}
  H_{\Mtilde}(u) - H_{\Mscr_{T+1}}(u) = h_k(u) > 2 \delta_C u.
\end{equation*}
Since $H_{\Mscr_{T+1}}(u) = u$ and $\delta_C \geq \delta(\Mtilde)$, it follows
that by Proposition~\ref{prop:fact-RH-approx} that $\Rtilde{\Mtilde} >
\Rtilde{\Mscr_{T+1}}$. This contradicts the assumption that the subroutine
terminates with $\Mscr_{T+1}$ because no further additions or exchanges result
in an increase of revenue. This finishes the proof of this case.

The proof of the lemma now follows.

\section{Summary and discussion} \label{sec:greedyopt_conclusion}

This paper focused on using choice models to make decisions. Assuming
that we have access to a revenue prediction subroutine, we designed an
algorithm to find an approximation of the optimal assortment with as
few calls to the revenue subroutine as possible.

We designed a general algorithm for the optimization of set-functions to solve
the static assortment optimization algorithms. Most existing algorithms (both
exact and approximate) heavily exploit the structure of the assumed choice
model; consequently, the existing algorithms -- even without any guarantees --
cannot be used with other choice models like the probit model or the mixture of
MNL models with a continuous mixture. Given these issues, we designed an
algorithm that is (a) not tailored to specific parametric structures and (b)
requires only a subroutine that gives revenue estimates for assortments. Our
algorithm is a sophisticated form of greedy algorithm, where the solution is
constructed from a smaller assortment through greedy additions and
exchanges. The algorithm is proved to find the optimal assortment exactly when
the underlying choice model is the MNL model. We also showed that the algorithm
is robust to errors in the revenue estimates provided by the revenue subroutine,
as long as the underlying choice model is the MNL model.

\bibliographystyle{plainnat}
\bibliography{ConcatBib}

\begin{thebibliography}{4}
\providecommand{\natexlab}[1]{#1}
\providecommand{\url}[1]{\texttt{#1}}
\expandafter\ifx\csname urlstyle\endcsname\relax
  \providecommand{\doi}[1]{doi: #1}\else
  \providecommand{\doi}{doi: \begingroup \urlstyle{rm}\Url}\fi

\bibitem[Rusmevichientong et~al.(2009)Rusmevichientong, Max~Shen, and
  Shmoys]{Rus09}
P.~Rusmevichientong, Z.J. Max~Shen, and D.B. Shmoys.
\newblock A ptas for capacitated sum-of-ratios optimization.
\newblock \emph{Operations Research Letters}, 37\penalty0 (4):\penalty0
  230--238, 2009.

\bibitem[Rusmevichientong et~al.(2010{\natexlab{a}})Rusmevichientong, Shen, and
  Shmoys]{Rus10}
P.~Rusmevichientong, Z.J.M. Shen, and D.B. Shmoys.
\newblock Dynamic assortment optimization with a multinomial logit choice model
  and capacity constraint.
\newblock \emph{Operations research}, 58\penalty0 (6):\penalty0 1666--1680,
  2010{\natexlab{a}}.

\bibitem[Rusmevichientong et~al.(2010{\natexlab{b}})Rusmevichientong, Shmoys,
  and Topaloglu]{Rus10mmnl}
P.~Rusmevichientong, D.~Shmoys, and H.~Topaloglu.
\newblock Assortment optimization with mixtures of logits.
\newblock Technical report, Tech. rep., School of IEOR, Cornell University,
  2010{\natexlab{b}}.

\bibitem[Talluri and van Ryzin(2004)]{Talluri04}
K.~Talluri and G.~J. van Ryzin.
\newblock Revenue management under a general discrete choice model of consumer
  behavior.
\newblock \emph{Management Science}, 50\penalty0 (1):\penalty0 15--33, 2004.

\end{thebibliography}

\end{document}